\documentclass[a4paper,11pt]{article}
\usepackage{graphicx}
\usepackage[margin=1.1in]{geometry}
\usepackage{amsmath,amssymb,amsfonts,amsthm,natbib,enumitem}
\usepackage{multirow}

\newcommand{\N}{\ensuremath{\mathbb{N}}}

\newcommand{\PP}{\ensuremath{\mathcal{P}}}
\newcommand{\E}{\ensuremath{\mathbb{E}}}
\newcommand{\I}{\ensuremath{\mathbb{I}}}

\newtheorem{theorem}{Theorem}
\newtheorem{lemma}{Lemma}

\newtheorem{example}{Example}

\newcommand{\ps}{${}^{\sim}$}

\newcommand{\Cn}{C(X_{1:n})}
\newcommand{\CnRL}{C_\text{RL}(X_{1:n})}
\newcommand{\CnS}{C_\text{S}(X_{1:n})}

\usepackage{array}
\newcolumntype{L}[1]{>{\raggedright\let\newline\\\arraybackslash\hspace{0pt}}m{#1}}
\newcolumntype{C}[1]{>{\centering\let\newline\\\arraybackslash\hspace{0pt}}m{#1}}
\newcolumntype{R}[1]{>{\raggedleft\let\newline\\\arraybackslash\hspace{0pt}}m{#1}}

\title{Implementing Monte Carlo Tests with P-value Buckets}
\author{Axel Gandy, Georg Hahn and Dong Ding
  \\Department of Mathematics, Imperial College London}
\date{}

\begin{document}
\maketitle

\begin{abstract}
Software packages usually report the results of statistical tests
using p-values. Users often interpret these by comparing them to
standard thresholds, e.g.\ 0.1\%, 1\% and 5\%, which is sometimes
reinforced by a star rating (***, **, *).
We consider an arbitrary statistical test whose p-value $p$ is not
available explicitly, but can be approximated by Monte Carlo
samples, e.g.\ by bootstrap or permutation tests. The standard
implementation of such tests usually draws a fixed number of samples
to approximate $p$. However, the probability that the exact and the
approximated p-value lie on different sides of a threshold
(the resampling risk) can be high, particularly for p-values
close to a threshold. We present a method to overcome this.
We consider a finite set of user-specified intervals which cover
$[0,1]$ and which can be overlapping. We call these p-value
buckets. We present algorithms that, with arbitrarily high probability,
return a p-value bucket containing $p$.
We prove that for both a bounded resampling risk and a finite runtime,
overlapping buckets need to be employed, and that our methods both bound
the resampling risk and guarantee a finite runtime for such overlapping buckets.
To interpret decisions with overlapping buckets, we propose an extension of the star rating system.
We demonstrate that our methods are suitable for use in standard software,
including for low p-value thresholds occurring in multiple testing settings,
and that they can be computationally more efficient than standard implementations.
\end{abstract}

\section{Introduction}
\label{section_introduction}
Software packages usually report the significance of statistical tests
using p-values.  The result of the test will often be interpreted by
comparing those p-values to thresholds. To facilitate this,
many tests in statistical software such as \textit{R}
\citep{Rsoftware}, \textit{SAS} \citep{SASsoftware} or \textit{SPSS}
\citep{SPSSsoftware} translate the significance to a star rating
system, in which typically $p\in (0.01,0.05]$ is denoted by *,
$p\in (0.001,0.01]$ is denoted by ** and $p\leq 0.001$ is denoted by ***.
As pointed out in the literature, such levels of significance are sensible since
they capture the magnitude of a p-value rather than its precise value,
which in contrast to the magnitude is usually not reliably estimated \citep{BoosStefanski2011}.

In this article, we are concerned with statistical tests whose p-value $p$ can only
be approximated by sequentially drawn Monte Carlo samples. Among others, this
scenario arises in bootstrap or permutation
tests \citep{Lourenco2014,Martinez2014,Liu2013,Wu2013,Asomaning2012,Dazard2012}.

Standard implementations of Monte Carlo tests in software packages
usually take a fixed number of samples and estimate $p$ as the
proportion of exceedances over the observed value of the test statistic.
Examples of this  include the computation of a bootstrap
p-value inside the function \textit{chisq.test} in \textit{R} or the
function \textit{t-test} in \textit{SPSS}. However, there is no control of the
\textit{resampling risk}, the probability that the exact and the
approximated p-value lie on two opposite sides of a testing threshold
(usually $0.1\%$, $1\%$ or $5\%$).

Sequential methods to approximate p-values have been studied in the
literature.  Early works provided ad hoc attempts to reduce the
computational effort without focusing on a specific error criterion
\citep{BesagClifford1991,Silva2009}.

Further developments aimed at a uniform bound on the resampling
risk for a single threshold
\citep{DavidsonMacKinnon2000,AndrewsBuchinsky2000,AndrewsBuchinsky2001,Gandy2009}.
\cite{Gandy2009} shows that such a uniform bound necessarily
results in an infinite runtime.

There are also approaches that aim to bound an integrated resampling risk for a single threshold
\citep{FayFollmann2002,Kim2010,SilvaAssuncao2013}.
Such an error criterion is weaker than a uniform bound on the resampling risk and can be
achieved with finite effort. 

In this article, we present algorithms that work with multiple thresholds, aim for uniform bounds on the resampling risk
and, under conditions, have a finite runtime.
We first generalize testing thresholds to a finite set of user-specified
intervals (called ``p-value buckets'') which cover $[0,1]$
and which can be overlapping.  Our algorithms return one of those p-value
buckets which is guaranteed to contain the unknown (true) $p$ up to a uniformly bounded error.

We prove that methods achieving both a finite runtime and a bounded resampling risk need to operate on overlapping p-value buckets.
In order to report decisions computed with overlapping buckets,
we propose to use an extension of the classical star rating system (*, **, ***) used to indicate the significance of a hypothesis.

Our methods rely on the computation of a confidence sequence for $p$,
i.e. a sequence of random intervals with a joint coverage probability.
We present two approaches to compute such a confidence sequence, prove
that both approaches indeed bound the resampling risk and achieve a
finite runtime for overlapping buckets.  We compare both approaches in
a simulation section and demonstrate that they achieve a competitive
computational effort which is close to a theoretical lower bound on
the effort we derive.

The article is structured as follows.
Section~\ref{sec:generalidea} introduces the mathematical setting of our article (Section~\ref{section_setting}),
the rationale behind overlapping p-value buckets (Section~\ref{sec:overlapping}),
our proposed extension of the traditional star rating system (Section~\ref{sec:starrating}) and
a general algorithm to compute a decision for $p$ with respect to a set of p-value buckets (Section~\ref{section_general_algorithm}).
The general algorithm relies on the construction of certain confidence sequences for $p$ for which we present two approaches:
one based on likelihood martingales \citep{Robbins1970,Lai1976} in Section~\ref{sec:robbinslai}
and one based on the \textit{Simctest} algorithm \citep{Gandy2009} in Section~\ref{sec:simctest}.
In Section~\ref{sec:effort} we first derive a theoretical lower bound on the expected effort (Section~\ref{sec:lowerbound})
and demonstrate that our methods achieve a computational effort which stays within a multiple of the optimal effort
(Sections~\ref{section_effortplots} and \ref{section_three_distributions}).
An application to multiple testing is considered in Section~\ref{sec:multiple_testing}.
The article concludes with a discussion in Section~\ref{section_discussion}.
All proofs can be found in Appendix~\ref{section_proofs}.
The Supplementary Material includes $R$ code to implement the algorithms as well as to reproduce all figures and tables. We have also implemented the method in the function {\it mctest} of the R-package {\it simctest}, which is available on CRAN.

\section{General algorithm}
\label{sec:generalidea}
\subsection{Setting}
\label{section_setting}
We consider one hypothesis $H_0$ which we would like to test with a
given statistical test.  Let $T$ denote the test statistic and let $t$
be the evaluation of $T$ on some given data. For simplicity, we assume
that $H_0$ should be rejected for large values of $t$.  In this case the p-value
is commonly computed  as  the probability of observing a
statistic at least as extreme as $t$, i.e.\
\begin{align}
p = \PP(T \geq t),
\label{eq:pvalue}
\end{align}
where $\PP$ is a probability measure under the null hypothesis. For
our purposes, we assume that $\PP$ is either the true distribution of
$T$ under $H_0$ or an estimate of it, e.g.\ the distribution implied
via bootstrapping.

We assume that the p-value $p$ is not available analytically but can
be approximated using Monte Carlo simulation by drawing independent
realizations of the test statistic $T$ under $\PP$. We will assume that
we can generate a sequence $X_i, i\in \N$, of those draws and we let
$X_i=1$ if the $i$th replicate is greater or equal than $t$ and
$X_i=0$ otherwise. As a consequence, $X_i$ have a Bernoulli$(p)$ distribution.

The algorithms we consider aim to return an interval containing $p$
from a given set ${\cal J}$ of possibly overlapping sub-intervals of $[0,1]$.
The algorithms are sequential; for $n=1,2,\dots$, based on 
$X_{1:n}=(X_1,\dots,X_n)$, they will decide if they can stop and return an interval or whether they need to observe more $X_i$.

We use $A$ to denote a generic algorithm of this type, $I_A$ (or
simply $I$) to denote the interval returned by $A$, and $\tau_A$ for
the stopping time of $A$, i.e.\ the number of $X_i$ that the
algorithm observes before returning $I_A$.
We use $J$ for generic elements of ${\cal J}$.
Our algorithms are built on the sequence $(X_i)$, which have $p$ as the unknown parameter.
To emphasize that the underlying distribution is determined by $p$,
we will use it as a subscript in the notation of probabilities and expected values, writing $\PP_p$ and $\E_p$.

Formally, we require ${\cal J}$ to be a \textit{set of p-value
buckets}, which we define to be a set of sub-intervals of $[0,1]$ of
positive length that cover $[0,1]$, i.e.\
$\bigcup_{J\in {\cal J}} J=[0,1]$.

For example,
\begin{equation}
  \label{eq:defJ0}
{\cal J}^0:=\{[0,10^{-3}],(10^{-3},0.01],(0.01,0.05],(0.05,1]\}
\end{equation}
is a set of p-value
buckets, which we will refer to in the remainder of the article as \textit{classical buckets}.
Deciding which of those
buckets $p$ falls into is equivalent
to deciding where $p$ lies in
relation to the three traditional
thresholds $0.001$, $0.01$ and $0.05$.

A natural error criterion for an algorithm $A$ is the risk of a
\textit{wrong} decision, defined as $\text{RR}_{p}(A)=\PP_p(p\notin I_A )$, and which we call the
\textit{resampling risk}. $\text{RR}_{p}(A)$ is a function of the p-value $p$.

The algorithms $A$ that we propose in this article  bound the
resampling risk uniformly in $p$ at a given $\epsilon \in (0,0.5)$, i.e.\
\begin{align}
  \text{RR}_p (A)\leq \epsilon \quad \text{ for all }p\in [0,1].
  \label{eq:RR}
\end{align}

\subsection{Overlapping buckets}
\label{sec:overlapping}
We say that the buckets $\cal J$ are \emph{overlapping} if for all
$p \in (0,1)$ there exists $J \in {\cal J}$ such that $p$ is contained
in the interior of $J$.  The following theorem shows that overlapping
buckets are both a necessary and sufficient prerequisite for a finite
time algorithm $A$ satisfying \eqref{eq:RR} to exist, where the effort is
measured in terms of the stopping time $\tau_{A}$.

\begin{theorem}
\label{theorem_existance}
The following statements are equivalent:
\begin{enumerate}
\setlength\itemsep{0em}
\item There exists an algorithm $A$ satisfying \eqref{eq:RR} with $\E_p(\tau_{A})<\infty$ for all $p \in [0,1]$.
  \item The p-value buckets $\cal J$  are overlapping.
  \item There exists an algorithm $A$ satisfying \eqref{eq:RR} with $\tau_{A}<C$ for some deterministic $C>0$.
  \end{enumerate}
\end{theorem}
All proofs can be found in Appendix~\ref{section_proofs}. A
consequence of the theorem is that  there is no algorithm $A$ with
finite expected effort (i.e., $\E_p(\tau_{A})<\infty$ for all $p \in [0,1]$) that achieves \eqref{eq:RR} for ${\cal J}^0$.

To turn the classical buckets ${\cal J}^0$ into a set of  overlapping p-value buckets,
we can add  intervals that contain the classical thresholds in their interior.
As a specific choice, we recommend
$${\cal J^{\ast}}={\cal J}^0 \cup \left\{ (5\times 10^{-4},2\times 10^{-3}],(0.008,0.012],(0.045,0.055] \right\}.$$
We will use ${\cal J}^{\ast}$ throughout the article and refer to them as the \textit{extended buckets}.
We recommend ${\cal J}^{\ast}$ for three reasons:
First, this choice results in roughly an equal maximal effort when $p$ is close to  all three classical thresholds (see Example~\ref{example_Jstar}).
Second, the maximal effort and the expected effort under the null are reasonable in practical applications.
Third, the interval limits in $\cal J^{\ast}$ have only few decimal places and can thus be easily written down.
Section~\ref{section_discussion} discusses additional (heuristic) ways of choosing buckets.

\subsection{Extended star rating system}
\label{sec:starrating}
It is commonplace to report the significance of a hypothesis using a star rating system:
strong significance is encoded as *** ($p<0.1\%$), significance at $1\%$ is encoded as ** and weak significance ($p<5\%$) as a single star.
This classification,
recommended in the publication manual of the American Psychological Association \citep[page 139]{ASA2010},
is the de facto standard for reporting significance.

\begin{table}[tbp]
\centering
\begin{tabular}{L{13mm}| cc cc cc cc}
\hline\hline
Bucket & \multicolumn{2}{C{3cm}}{$[0,0.1\%]$} & \multicolumn{2}{C{3cm}}{$(0.1\%,1\%]$} & \multicolumn{2}{C{3cm}}{$(1\%,5\%]$} & \multicolumn{2}{C{3cm}}{$(5\%,1]$}\\
Code & \multicolumn{2}{C{3cm}}{***} & \multicolumn{2}{C{3cm}}{**} & \multicolumn{2}{C{3cm}}{*}\\
\hline
Bucket & \multicolumn{1}{C{13mm}}{} & \multicolumn{2}{C{3cm}}{$(0.05\%,0.2\%]$} & \multicolumn{2}{C{3cm}}{$(0.8\%,1.2\%]$} & \multicolumn{2}{C{3cm}}{$(4.5\%,5.5\%]$}\\
Code & \multicolumn{1}{C{13mm}}{} & \multicolumn{2}{C{3cm}}{**\ps} & \multicolumn{2}{C{3cm}}{*\ps} & \multicolumn{2}{C{3cm}}{\ps}\\
\hline\hline
\end{tabular}
\caption{Extended star rating system for $\cal J^\ast$.\label{tab:report}}
\end{table}

We propose to extend the star rating system for the overlapping buckets
in ${\cal J}^{\ast}$ in the manner given in Table~\ref{tab:report} (referred to as the \textit{extended star rating system}).
The same coding could be used for other p-value buckets that contain
${\cal J}^0$.  If the p-value bucket $I$ returned by our algorithm
allows for a clear decision with respect to the classical thresholds
(first row of Table~\ref{tab:report}), we report the classical star
rating.  Otherwise, we propose to report significance with respect to
the smallest classical threshold larger than $\max I$ and to indicate
the possibility of a higher significance with a tilde symbol (second
row of Table~\ref{tab:report}).

For instance,
suppose an algorithm returns the bucket $I=(0.05\%,0.2\%]$ for $p$ upon stopping.
This implies  $p \leq 1\%$ and thus we can  safely report a ** significance.
However,
as $p$ could either be smaller or larger than the next classical threshold $0.1\%$,
we report **\ps~to indicate the possibility of a higher significance.

\subsection{The general construction}
\label{section_general_algorithm}
We suppose that  we can compute a confidence sequence  $\Cn$, $n\in \N$,
for $p$, i.e.\ a sequence of intervals  $\Cn$ such that its joint
coverage probability is at least
$1-\epsilon$, where $\epsilon>0$ is the desired uniform bound on the resampling
risk. Formally, we require
\begin{align}
\PP_p(p\in \Cn~\text{for all}~n \in \N)\geq 1-\epsilon \quad \text{for all } p\in [0,1].
\label{eq:covprob}
\end{align}
In Sections~\ref{sec:robbinslai} and~\ref{sec:simctest} we consider two constructions satisfying \eqref{eq:covprob}.

The generic algorithm we propose will depend on the choice of
$p$-value buckets ${\cal J}$ and the method $C$ for computing a
confidence sequence. We will denote the algorithm by $A({\cal J}, C)$.
We define the stopping time
\begin{align}
\tau_{A({\cal J}, C)}=\inf \left\{ n\in \N: \text{there exists}~J \in {\cal J}~\text{such that}~\Cn\subseteq J\right\}
\label{eq:stopping_rule}
\end{align}
which denotes the minimal number of samples $n$ needed until a
confidence interval $\Cn$ is fully contained in a bucket $J \in {\cal J}$.
If $\tau_{A({\cal J}, C)}<\infty$, the result of our algorithm is a bucket $I\in {\cal J}$ such that $\Cn \subseteq I$.
If multiple buckets exist with this property then an arbitrary one is chosen.
If $\tau_{A({\cal J}, C)}=\infty$, our algorithm returns an arbitrary element $I \in {\cal J}$ such that $\lim_{n\to\infty}\frac{S_n}{n} \in I$,
where $S_n = \sum_{i=1}^n X_i$. The limit exists by the law of large numbers. 

\begin{theorem}
  If \eqref{eq:covprob} holds then $A=A({\cal J}, C)$ satisfies \eqref{eq:RR}.
\end{theorem}

This is an immediate consequence of the construction and  the strong law of large numbers.

The confidence interval $\Cn$ for $p$ and the bucket $I \in {\cal J}$ that our algorithm returns are related but not equivalent.
Following \cite{BoosStefanski2011}, we are ultimately only interested in reporting one of the pre-specified p-value buckets that $p$ falls in;
a more precise confidence statement on $p$ is not required.
The confidence interval $\Cn$ for $p$ serves to quantify the uncertainty in the estimation of $p$,
and since $\Cn \subseteq I$ it ensures that the bucket $I$ we report satisfies \eqref{eq:RR}.

Lastly, if there exists $N\in \N$ such that $\tau_{A({\cal J}, C)}<N$, we can relax \eqref{eq:covprob} to
\begin{equation}
  \label{eq:covprobfinite}
  \PP_p(p\in \Cn~\text{for all}~n < N)\geq 1-\epsilon
  \quad \text{for all } p\in [0,1].
\end{equation}

\begin{figure}[tbp]
\includegraphics[width=0.49\textwidth]{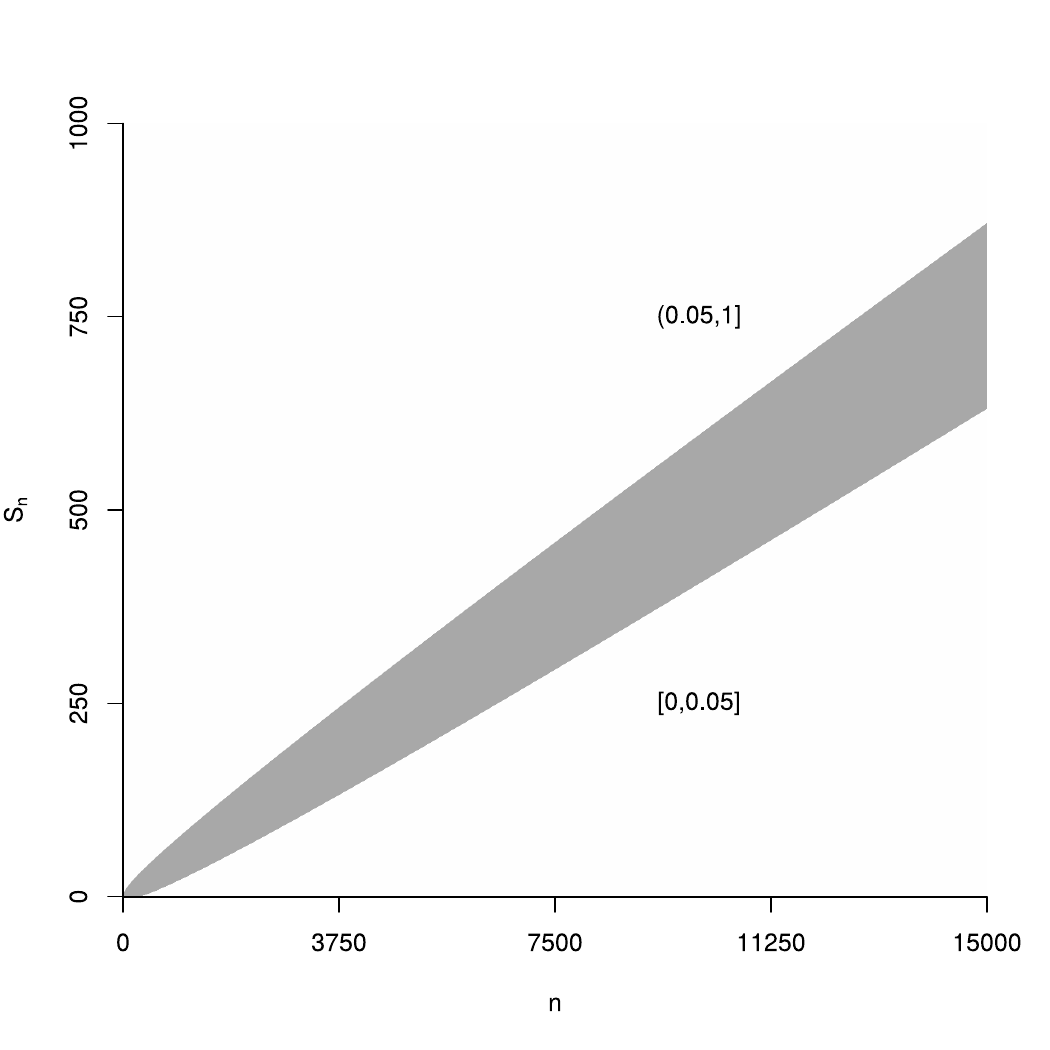}\hfill
\includegraphics[width=0.49\textwidth]{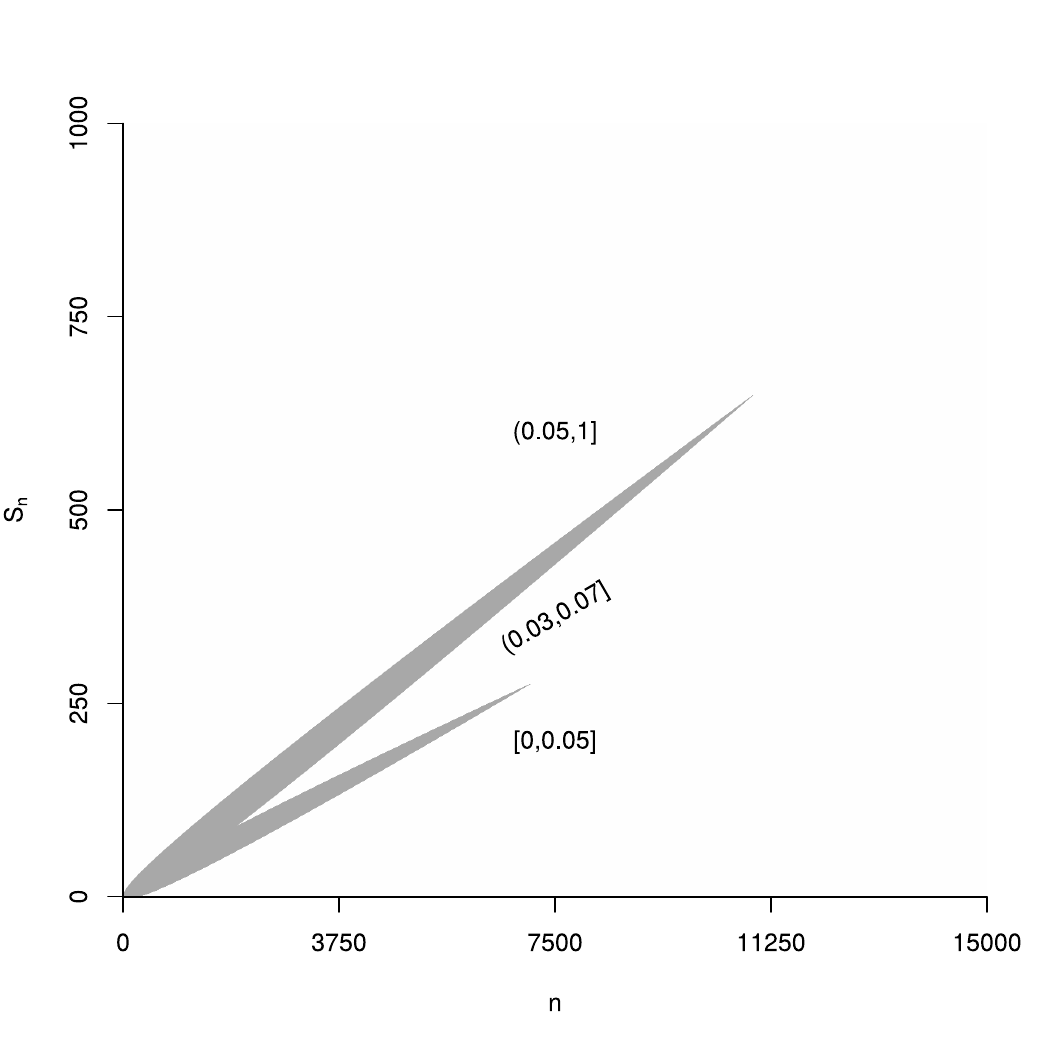}
\caption{Non-stopping region (grey) to compute a decision on $p$ with respect to $J^e$ (left), which corresponds to a $5\%$ threshold,
and with respect to the overlapping buckets $J^e \cup \{ (0.03,0.07] \}$ (right).\label{fig:basicidea}}
\end{figure}

\begin{example}
\label{example_basicidea}
Suppose we are solely interested in the 5\% threshold.
Testing at 5\% corresponds to the two classical buckets ${\cal J}^e = \{ [0,0.05],(0.05,1] \}$.
Using the approach of Section~\ref{sec:simctest} with $\epsilon=10^{-3}$ to compute a confidence sequence for $p$,
we arrive at the non-stopping region displayed in Figure~\ref{fig:basicidea} (left).
We define the non-stopping region as the region in which sampling progresses until the sampling path $(n,S_n)$ hits either its lower or upper boundary.
As displayed in Figure~\ref{fig:basicidea} (left),
we report the interval $[0,0.05]$ $\left( (0.05,1] \right)$ upon hitting the lower (upper) boundary first.

Adding the bucket $(0.03,0.07]$ to ${\cal J}^e$ results in overlapping buckets
with a finite non-stopping region displayed in Figure~\ref{fig:basicidea} (right).
In Figure~\ref{fig:basicidea} (right), the sample path can leave the non-stopping region in three ways:
Either to the top via the former upper boundary of Figure~\ref{fig:basicidea} (left),
in which case we report the classic interval $(0.05,1]$,
to the bottom via the former lower boundary corresponding to the bucket $[0,0.05]$,
or to the middle corresponding to the added bucket $(0.03,0.07]$.
\end{example}

\begin{figure}[tbp]
\centering
\includegraphics[width=0.49\textwidth]{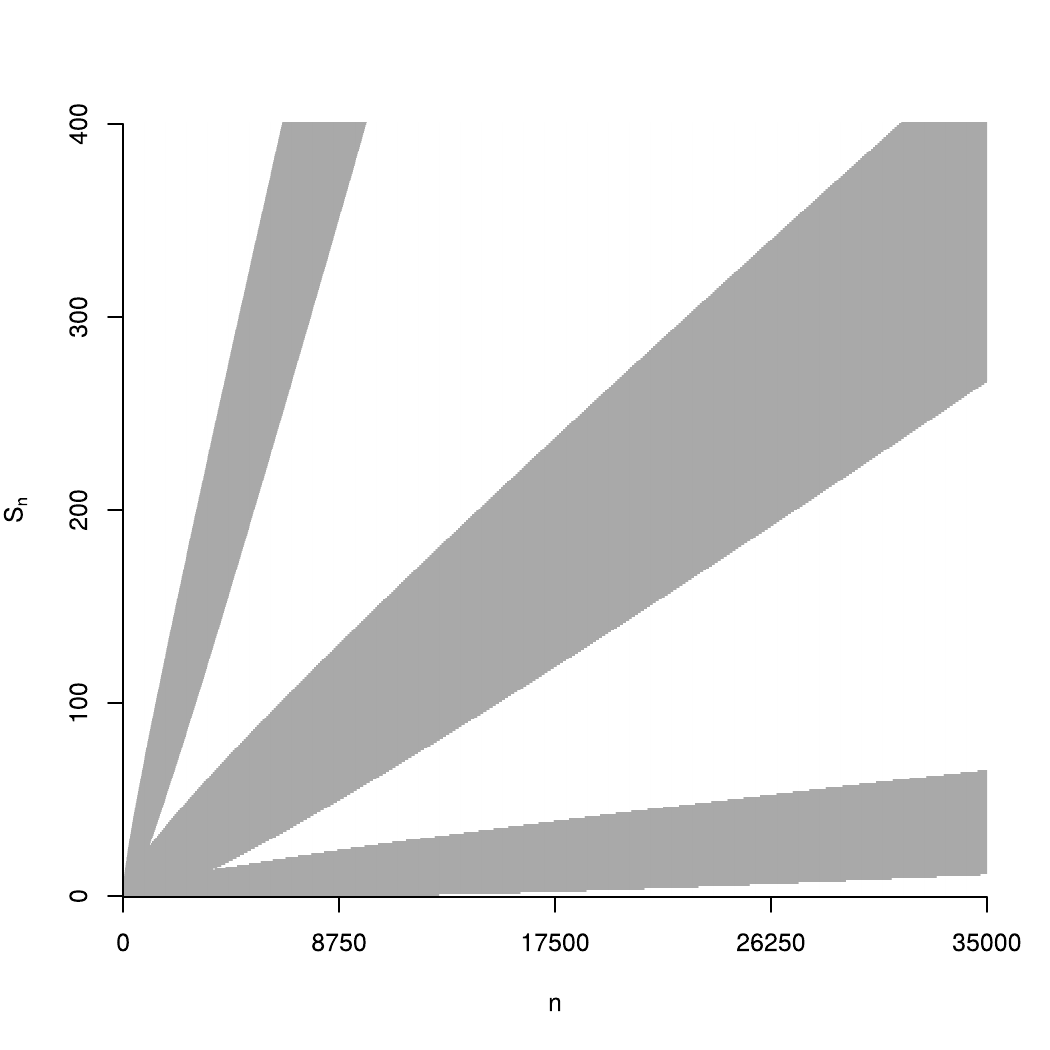}\hfill
\includegraphics[width=0.49\textwidth]{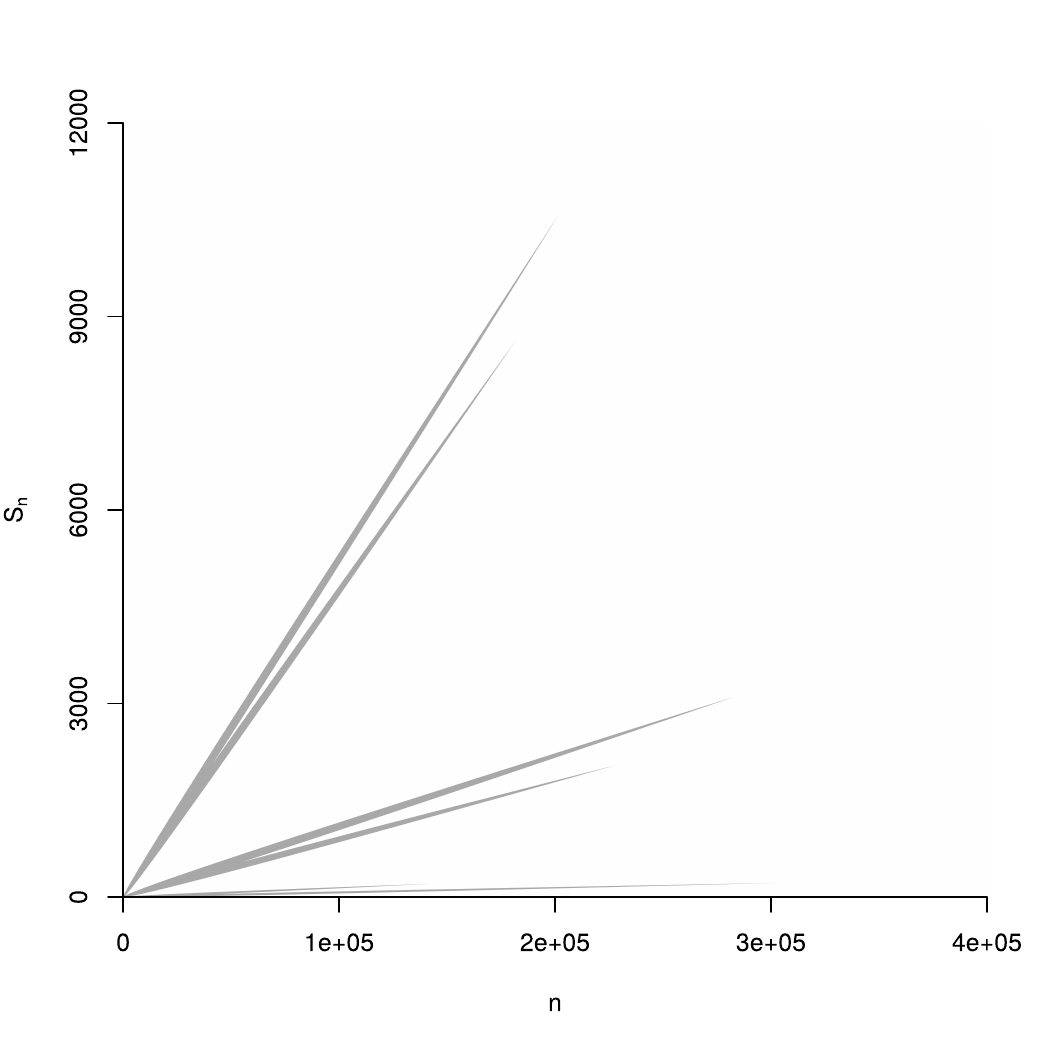}
\caption{Non-stopping region (gray) for ${\cal J}^0$ (left) and ${\cal J}^\ast$ (right).\label{fig:stoppingregion}}
\end{figure}

\begin{figure}[tbp]
\centering
\includegraphics[width=0.45\textwidth]{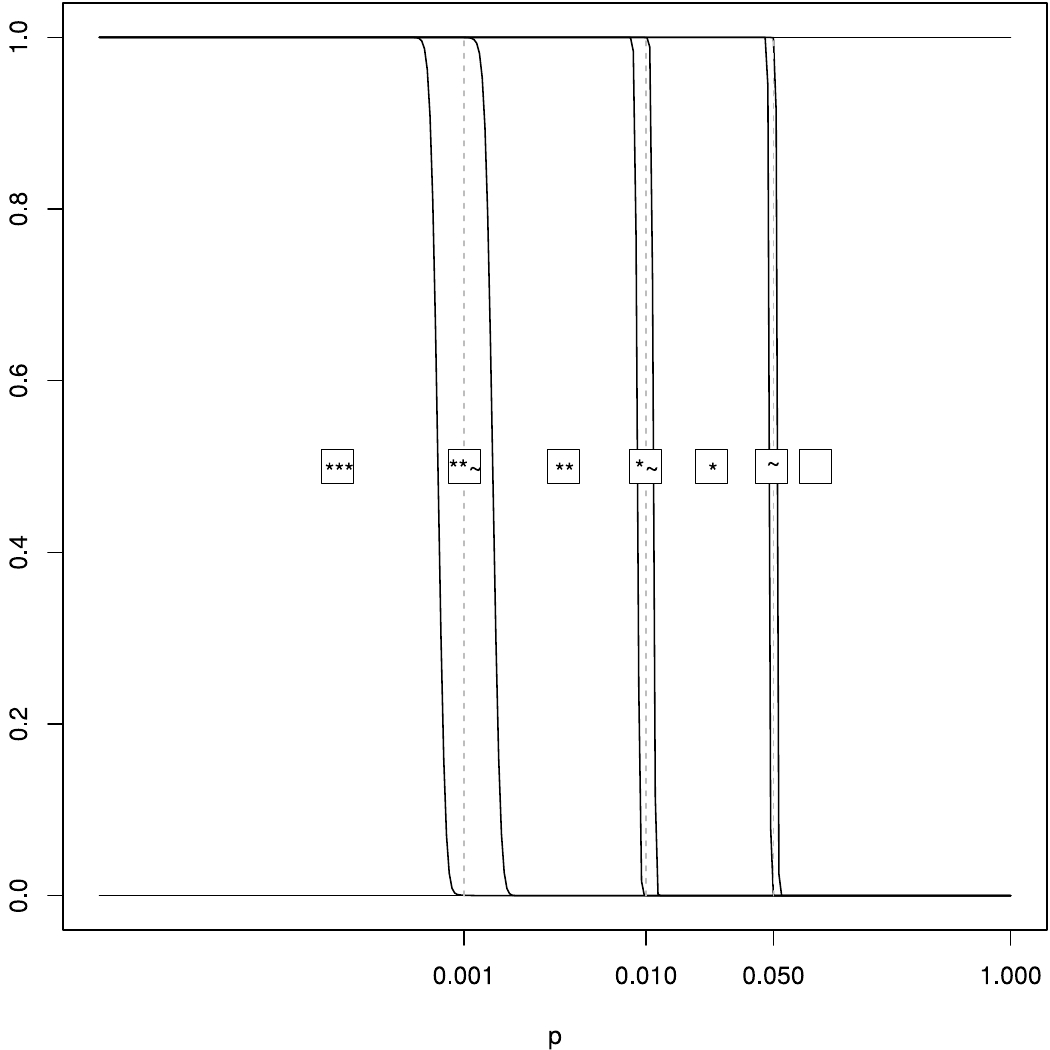}
\caption{Probabilities of observing each possible decision for $\cal J^\ast$ as a function of $p$.\label{fig:outcomes}}
\end{figure}

\begin{example}
\label{example_Jstar}
Similarly to Example~\ref{example_basicidea},
Figure~\ref{fig:stoppingregion} shows the non-stopping region for
${\cal J}^0$ and $\cal J^\ast$. The stopping region is infinite
for the non-overlapping ${\cal J}^{0}$ and finite for the overlapping buckets ${\cal J}^{\ast}$.

How likely is it to observe the different decisions which can occur when testing with ${\cal J}^\ast$?
Figure~\ref{fig:outcomes} shows the probability of obtaining each decision in the extended star rating
system for ${\cal J}^\ast$ as a function of $p$.  These probabilities are computed as
follows: For a given $p$, we iteratively (over $n$) compute the
distribution of $S_{n}$ conditional on not stopping. This allows us to
compute the probability of stopping
and the resulting decision.

Figure~\ref{fig:outcomes} shows that intermediate decisions (\ps, *\ps, **\ps)
only occur with appreciable probability for a narrow range of p-values.
For most p-values, a decision in the sense of the classical star rating system is reached.
\end{example}

\section{Construction of confidence sequences}
\label{sec:confseq}
We now present two approaches for computing confidences sequences and show that, for overlapping buckets, the resulting stopping times are bounded.

\subsection{The Robbins-Lai approach}
\label{sec:robbinslai}
\cite{Robbins1970} showed that the sequence of sets
$$\CnRL = \{ p \in [0,1]: (n+1)b(n,p,S_n)>\epsilon\}$$
satisfies \eqref{eq:covprob}, where $b(n,p,s) = \binom{n}{s} p^{s} (1-p)^{n-s}$ (see eq.~\eqref{eq:conf}).
\cite{Lai1976} showed that $\CnRL$ are intervals.
Using these intervals with overlapping buckets leads to a bounded effort:
\begin{lemma}
  \label{le:finitestop}
  If ${\cal J}$ are overlapping buckets then the stopping time
  $\tau_{A({\cal J}, C_\text{RL})}$ can be bounded by a deterministic positive constant.
\end{lemma}

The intervals $\CnRL$ need not be computed explicitly in order to check \eqref{eq:stopping_rule}.
Appendix~\ref{section_simple_criterion_RL} gives a simple criterion to check if $\CnRL \subseteq J$ for  $J \in {\cal J}$.

\subsection{The Simctest approach}
\label{sec:simctest}
\cite{Gandy2009} provides a method to compute a decision for $H_0$
with respect to a single threshold in the same Monte Carlo setting as the one of Section~\ref{section_setting}.
This approach can also be used to construct confidence sequences for multiple thresholds.

For the purposes of this article, it suffices to mention that for a given threshold $\alpha\in [0,1]$,
\cite{Gandy2009} constructs two integer valued stopping boundaries $(L_{n,\alpha})_{n\in \N}$ and $(U_{n,\alpha})_{n\in \N}$, and defines a stopping time
$$\tau_{\alpha} = \inf \{k \in \N: S_k \geq U_{k,\alpha} \text{ or } S_k \leq L_{k,\alpha} \}.$$
The construction is parametrized by a spending sequence $(\epsilon_n)_{n\in\N}$ that is nonnegative, nondecreasing and converges to some $0<\rho<1$.
 \cite[Theorem 1]{Gandy2009}  shows that, under conditions,
$\E_p(\tau_{\alpha})<\infty$ for $p\neq \alpha$ and that the probability of hitting the wrong
boundary is bounded by $\rho$, i.e.\ $\PP_p(S_{\tau_{\alpha}} \geq U_{\tau_{\alpha},\alpha})<\rho$ for $p<\alpha$, and similarly for $p>\alpha$.

In order to extend this approach to multiple thresholds,
we first define the set of boundaries of intervals in $\cal J$ that are in the interior of $[0,1]$:
$$B_{\cal J} = \{\min J,~\max J:~J \in {\cal J}\}\setminus \{0,1\}.$$
Then, for each $\alpha \in B_{\cal J}$ we construct the stopping boundaries $L_{n,\alpha}$ and $U_{n,\alpha}$ using the same $\rho$.
We define
$$I_{n,\alpha}=
  \begin{cases}
    [0,1]&\text{if }n<\tau_{\alpha},\\
    [0,\alpha)&\text{if }n\geq\tau_{\alpha}, S_{\tau_{\alpha}}\leq  L_{\tau_{\alpha},\alpha},\\
    (\alpha,1]&\text{if }n\geq\tau_{\alpha}, S_{\tau_{\alpha}}\geq  U_{\tau_{\alpha},\alpha}.
  \end{cases}$$
We define the confidence sequence of the Simctest approach as $\CnS=\bigcap_{\alpha \in B_{\cal J}} I_{n,\alpha}$.

The following theorem shows that $\CnS$ has the desired joint coverage probability given in \eqref{eq:covprob}
(or \eqref{eq:covprobfinite} for overlapping buckets) when setting $\rho=\epsilon/2$.
Moreover, the theorem shows that the algorithm $A({\cal J},C_\text{S})$ has a bounded stopping time if $\cal J$ is a finite set of overlapping buckets.

\begin{theorem}
\label{theorem_2epsilonNEW}
Let $\epsilon\in (0,1)$. For each $\alpha \in B_{\cal J}$, construct $L_{n,\alpha}$ and $U_{n,\alpha}$ with error probability $\rho=\epsilon/2$.
Let $N\in \N\cup \{\infty\}$.
Suppose that $U_{n,\alpha} \leq U_{n,\alpha'}$ and $L_{n,\alpha} \leq L_{n,\alpha'}$ for all $\alpha, \alpha' \in B_{\cal J}$, $\alpha < \alpha'$, and $n < N$.

\begin{enumerate}
  \item \label{thm_epsilon_part1}
  Then $\PP_p(p\in \CnS \text{ for all } n < N) \geq 1-\epsilon \quad  \text{for all }p\in [0,1].$
  \item \label{thm_epsilon_part2} Suppose $N=\infty$,  $\rho \leq 1/4$ and $\log(\epsilon_n - \epsilon_{n-1}) = o(n)$ as $n \rightarrow \infty$.
  If ${\cal J}$ is a finite set of overlapping $p$-value buckets then there exists $c<\infty$ such that $\tau_{A({\cal J},C_\text{S})} \leq c$.
\end{enumerate}
\end{theorem}
Allowing $N<\infty$ in Theorem~\ref{theorem_2epsilonNEW} is useful for stopping boundaries constructed to yield a finite runtime (see \eqref{eq:covprobfinite}).

The condition on the spending sequence in part~\ref{thm_epsilon_part2} of Theorem~\ref{theorem_2epsilonNEW}
is identical to the condition imposed in Theorem~1 of \cite{Gandy2009}.
It is satisfied by the \textit{default spending sequence} defined in \cite{Gandy2009} as $\epsilon_n=\rho n/(n+k)$ with $k=1000$,
which is also employed in the remainder of this article.

The condition on the monotonicity of the boundaries
($U_{n,\alpha} \leq U_{n,\alpha'}$ and $L_{n,\alpha} \leq L_{n,\alpha'}$
for all $n \in \N$ and $\alpha, \alpha' \in B_{\cal J}$ with $\alpha < \alpha'$)
can be checked for a fixed spending sequence $(\epsilon_n)_{n \in \N}$ in two ways:
For finite $N$, the two inequalities can be checked manually after constructing the boundaries.
For $N=\infty$, the following lemma shows that under conditions, the
monotonicity of the boundaries holds true for all $n\geq n_0$, where
$n_0 \in \N$ can be computed as a solution to inequality
\eqref{eq:def_no} given in the proof of Lemma~\ref{le:eventual_ineqbounds} in Appendix~\ref{section_proofs}.

\begin{lemma}
\label{le:eventual_ineqbounds}
Suppose $\rho \leq 1/4$ and $\log(\epsilon_n - \epsilon_{n-1}) = o(n)$ as $n \rightarrow \infty$.
Let $\alpha,\alpha' \in B_{\cal J}$ with  $\alpha < \alpha'$.
Then there exists $n_{0} \in \mathbb{N}$ such that for all $n \geq n_0$,
$$L_{n,\alpha} \leq L_{n,\alpha'} \quad \text{and} \quad U_{n,\alpha} \leq U_{n,\alpha'}.$$
\end{lemma}
For $n<n_0$, the inequalities again have to be checked manually.

\section{Computational effort}
\label{sec:effort}
This section investigates the expected computational effort of the
algorithm of Section~\ref{section_general_algorithm}.  We start by
deriving a theoretical lower bound on the expected effort in
Section~\ref{sec:lowerbound}.  We then compare both the Simctest and
Robbins-Lai approach of Section~\ref{sec:confseq} in terms of their
expected effort as a function of $p$
(Section~\ref{section_effortplots}).  Integrating this effort for
certain p-value distributions of practical interest allows us to compare
both approaches in practical situations
(Section~\ref{section_three_distributions}).  Section
\ref{sec:multiple_testing} shows that the algorithm can be used for
small p-values arising in multiple testing settings.

\subsection{Lower bounds on the expected effort}
\label{sec:lowerbound}

In this section we construct lower bounds on the expected number of
steps of sequential procedures satisfying \eqref{eq:RR}. The key idea is to consider hypothesis tests  implied
by \eqref{eq:RR} and then to use the lower bounds for the expected
effort of sequential tests \cite[eq.~(4.80)]{wald1945}.

\begin{theorem}
\label{th:lowerbounds}
Let $\tau$ be the number of steps taken by a sequential procedure returning  $I\in {\cal J}$ which respects \eqref{eq:RR}.
Then, for every $\tilde p\in [0,1]$,
\begin{equation}
  \label{eq:lowerbound1}
  \E_{\tilde p}(\tau)\geq\sup_{q\notin \tilde J}  e(\tilde p, q, \epsilon,\epsilon),
\end{equation}
where $\tilde J=\bigcup_{J\in {\cal J}, \tilde p\in J}J$ is the
union of all buckets containing $\tilde p$ and
$$e(p,q,\alpha,\beta)=\frac{(1-\alpha)\log(\beta/(1-\alpha))+\alpha\log((1-\beta)/\alpha)}{p\log(q/p)+ (1-p)\log((1-q)/(1-p))}.$$
Furthermore, if $\tilde p\in [0,1]$ is such that exactly two  elements of ${\cal J}$
contain $\tilde p$, say $J_1$ and $J_2$, then
\begin{equation}
  \label{eq:lowerbound2}
  \E_{\tilde p}(\tau) \geq  \min_{\eta\in [0,1]} \max\left\{ \sup_{q\notin J_1}e(\tilde p, q,1-\eta,\epsilon),\sup_{q\notin J_2}e(\tilde p, q,\min(\eta+\epsilon,1),\epsilon) \right\} .
\end{equation}
\end{theorem}

We call the bound given by \eqref{eq:lowerbound1} the basic lower
bound and the bound given by the maximum of \eqref{eq:lowerbound1} and
\eqref{eq:lowerbound2} the improved lower bound.
The suprema  in \eqref{eq:lowerbound1} and \eqref{eq:lowerbound2} can be evaluated by looking at
the boundary points of $\tilde J$, $J_1$ and $J_2$. The minimum can
be bounded from below by looking at a grid of values for
$\eta$ and by conservatively replacing
$e(\tilde p, q,\eta+\epsilon,\epsilon)$ by $e(\tilde p, q,\eta+\epsilon+\delta,\epsilon)$, where $\delta$ is the grid
width. This is because $e$ is decreasing in its third argument.

\begin{figure}[tbp]
\centering
\includegraphics[width=0.65\textwidth]{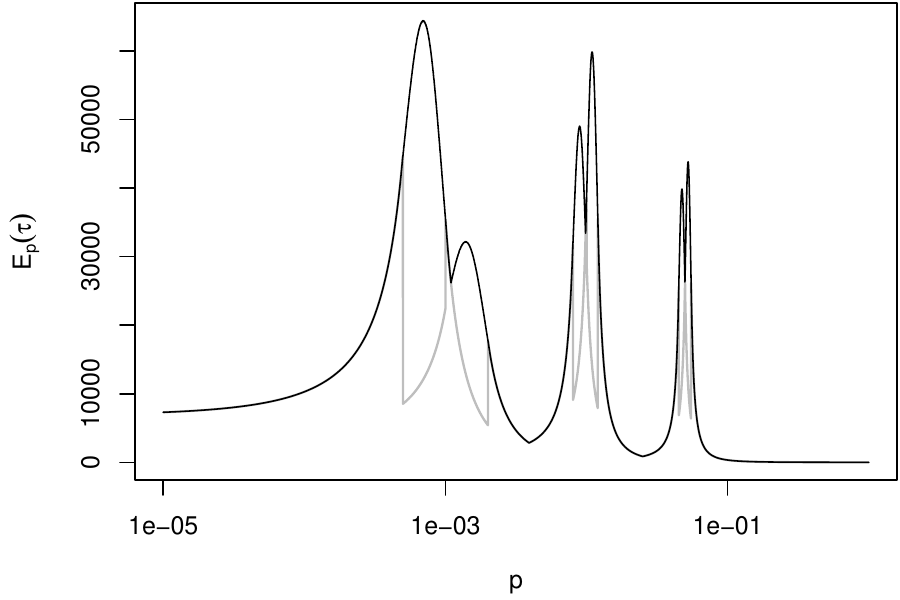}
\caption{Basic (grey) and improved (black) lower bounds on the effort $\E_p(\tau)$ for ${\cal J}^\ast$.\label{fig:lowerboundexample}}
\end{figure}

Figure~\ref{fig:lowerboundexample}
gives an example of both the basic and the improved lower bounds on $\E_{\tilde p}(\tau)$ for the extended buckets ${\cal J}^\ast$.
The improved bound is much higher (and thus better) in the areas where there are overlapping buckets.

\subsection{Expected effort for (non-)overlapping buckets}
\label{section_effortplots}

This section investigates both the classical buckets $\cal J$
as well as the extended buckets ${\cal J}^\ast$
with respect to the implied expected effort as a function of $p$.

\begin{figure}[tbp]
\includegraphics[width=0.49\textwidth]{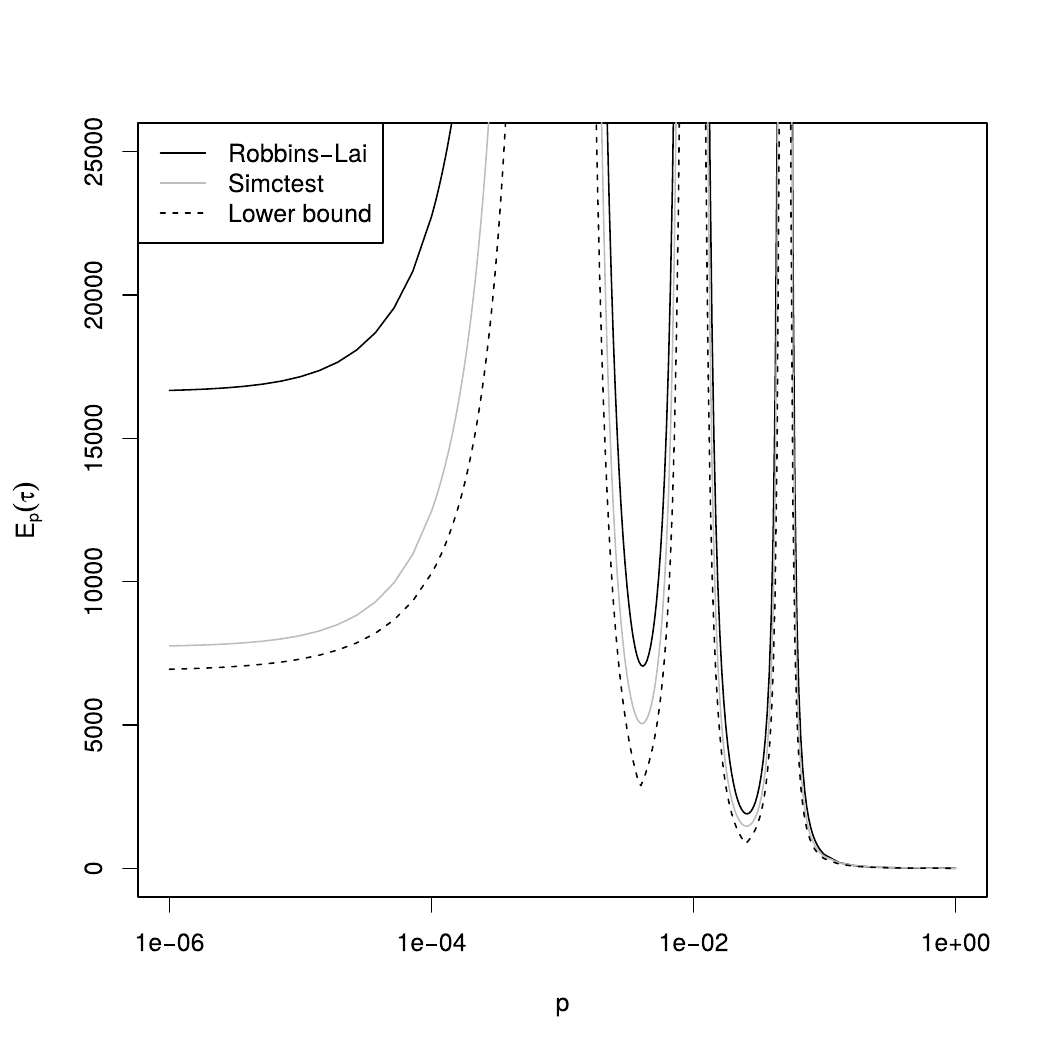}\hfill
\includegraphics[width=0.49\textwidth]{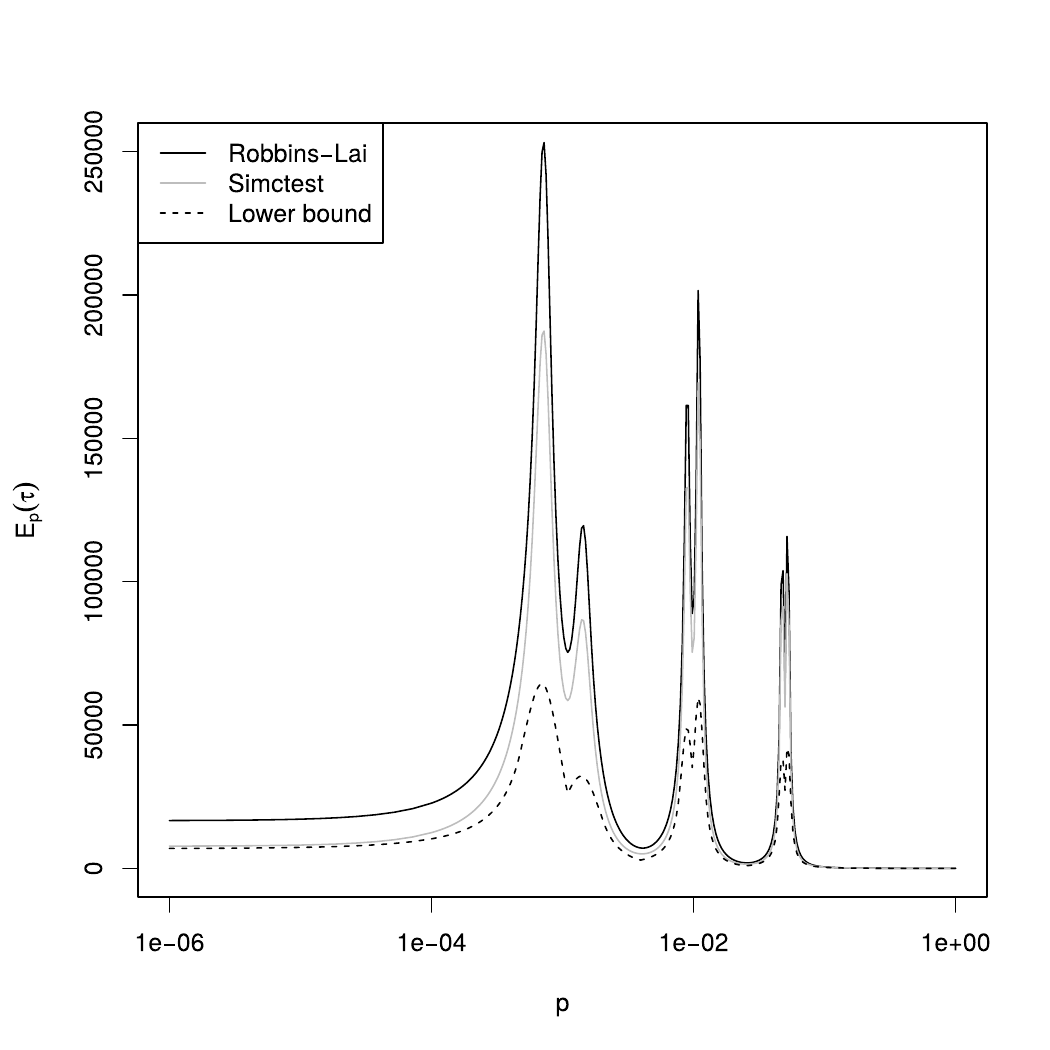}
\caption{Expected effort to compute a decision with respect to ${\cal J}^0$ (left) and ${\cal J}^\ast$ (right) as a function of $p$.
Confidence sequences computed with both Simctest (grey) and Robbins-Lai (black).
Lower bound on the effort indicated with a dashed line.\label{fig:effort}}
\end{figure}

Using the non-stopping regions depicted in Figure~\ref{fig:stoppingregion},
Figure~\ref{fig:effort} shows the expected effort (measured in terms of the number of samples drawn)
to compute a decision with respect to $\cal J$ (left) and ${\cal J}^\ast$ (right) as a function of $p \in [10^{-6},1]$.
For any given $p$, the expected effort is computed
by iteratively (over $n$) updating the distribution of $S_n$ conditional on not having stopped up to time $n$.
Using this distribution,
we work out the probability of stopping at step $n$ and add the appropriate contribution to the overall effort.
For both the Robbins-Lai and the Simctest approach,
the files \textit{RL.cpp} and \textit{simctest.cpp} included in the Supplementary Material
contain an implementation that computes the effort for a fixed set of p-value buckets.

The effort diverges as $p$ approaches any of the thresholds in $\cal J$.
For $\cal J^\ast$ the effort stays finite even in the case that $p$ coincides with one of the thresholds (Figure~\ref{fig:effort}, right).
The effort is maximal in a neighborhood around each threshold, while in-between thresholds, the effort slightly decreases.
For p-values larger than the maximal threshold in $\cal J$ or $\cal J^\ast$ the effort decreases to zero.
The effort for Simctest seems to be uniformly smaller than the one for Robbins-Lai for both $\cal J$ and ${\cal J}^\ast$.

Figure~\ref{fig:effort} also shows the lower bound (dashed line) on the effort derived in Section~\ref{sec:lowerbound}.
Using Simctest,
the effort of our algorithm of Section~\ref{section_general_algorithm}
differs from the theoretical lower bound by only a small factor.

\subsection{Expected effort for three specific p-value distributions}
\label{section_three_distributions}

The expected effort  of the proposed methods for repeated use
can be obtained by integrating the expected effort for a fixed $p$
(see Figure~\ref{fig:effort}, right) with respect to certain p-value distributions.

Here, we consider using the extended buckets ${\cal J}^{\ast }$
with three different p-value distributions.  These are a uniform
distribution in the interval $[0,1]$ ($H_0$), as well as two
alternatives given by the density $\frac{1}{2}+10\I(x \leq 0.05)$
($H_{1a}$) and by a Beta$(0.5,25)$ distribution ($H_{1b}$), where $\I$
denotes the indicator function.

\begin{table}[tbp]
\centering
\begin{tabular}{lrrr}
  \hline\hline
 & Robbins-Lai & Simctest & Lower bound \\ 
  \hline
  $H_0$ & 2228 & 1853 & 975 \\ 
  $H_{1a}$ & 16878 & 13837 & 7126 \\ 
  $H_{1b}$ & 40059 & 30896 & 15885 \\ 
   \hline\hline
\end{tabular}
\caption{Expected (integrated) effort for both Robbins-Lai and Simctest applied to ${\cal J^\ast}$.\label{tab:effort}}
\end{table}

Table~\ref{tab:effort} shows the expected effort as well as the lower bound on the expected effort.
The Simctest approach (Section~\ref{sec:simctest})
dominates the one of Robbins-Lai (Section~\ref{sec:robbinslai})
for this specific choice of distributions.
As expected, the effort is lowest for a uniform p-value distribution,
and more extreme for the alternatives having higher probability mass on low p-values.
Using Simctest,
the expected effort stays within roughly a factor of two of the theoretical lower bound derived in Section~\ref{sec:lowerbound}.

\subsection{Application to multiple testing}
\label{sec:multiple_testing}
We consider the applicability of our algorithm of Section~\ref{section_general_algorithm} to the (lower) testing thresholds
occurring in multiple testing scenarios.
In the following example, we demonstrate that our algorithm is well suited as a screening procedure for the most significant hypotheses.
Even for small threshold values, it is capable of detecting more rejections than a na\"ive sampling procedure that uses an equal number of samples for each hypothesis.

We assume we want to test $n=10^4$ hypotheses using the \cite{Bonferroni1936} correction to correct for multiplicity.
In order to be able to compute numbers of false classifications,
we assign $n_\text{alt}=100$ hypotheses to the alternative, the remaining $n-n_\text{alt}=9900$ hypotheses are from the null.
The p-values of the alternative are then set to $1-F(X)$,
where $F$ is the cumulative distribution function of a Student's $t$-distribution with $100$ degrees of freedom
and $X$ is a random variable sampled from a $t$-distribution with $100$ degrees of freedom
and noncentrality parameter uniformly chosen in $[2,6]$.
The p-values of the null are sampled uniformly in $[0,1]$.

In order to screen hypotheses, we aim to group them by the order of magnitude of their p-values.
For this we employ the overlapping buckets
$${\cal J}^s = \left\{ \left[0,10^{-7}\right] \right\} \cup \left\{ \left(10^{i-2},10^{i}\right]: i = -6,\ldots,0 \right\}$$
which group the p-values in buckets spanning two orders of magnitude each (and  $\left[0,10^{-7}\right]$).

We apply our algorithm $A({\cal J}^s,C_\text{S})$ of Section~\ref{section_general_algorithm} to ${\cal J}^s$
using confidence sequences computed with the Simctest approach (Section~\ref{sec:simctest}) and parameter $\epsilon=10^{-3}$.
To speed up the Monte Carlo sampling,
we sample in batches of geometrically increasing size $\lfloor a^i b \rfloor$ in each iteration $i \in \N$,
where $b=10$ and $a=1.1$.
Likewise, both the stopping boundaries and the stopping condition (hitting of either boundary) in Simctest are updated and checked in batches of the same size.

We now report the results from a single run of this setup. Our algorithm draws $N=3.2 \times 10^5$ samples per hypothesis.
Of the $10^4$ hypotheses, $28$ are correctly allocated to the two lowest buckets.
As expected, the p-values from the null are all allocated to larger buckets (covering values from $10^{-4}$ onwards).

An alternative approach would be to draw an equal number of $N$ samples per hypothesis
and to compute a p-value using a pseudo-count \citep{davison1997bootstrap}.
Due to this pseudo-count, this na\"ive approach is incapable of observing p-values below $(N+1)^{-1} = 3.125 \times 10^{-6}$ (see also \cite{GandyHahn2017}),
and in particular incapable of observing any p-values in the two lowest buckets.

\section{Discussion}
\label{section_discussion}
The overlapping p-value buckets presented in Section~\ref{sec:overlapping}
were chosen to be easily written down and to yield an equal maximal effort for all classical thresholds as well as a reasonable expected effort.
However, these criteria are essentially arbitrary.
A variety of further (heuristic) criteria can be used to obtain overlapping buckets from traditional testing thresholds $T=\{ t_0,\ldots,t_m \}$.
These include:
\begin{enumerate}
  \setlength\itemsep{1em}
  \item The bucket overlapping each threshold $t \in T$ can be chosen as $[\rho t, \rho^{-1}t]$ for a fixed proportion $\rho \in (0,1)$.
  \item Since the length of a confidence interval for a binomial quantity (with success probability $p$) behaves proportionally to
  $\sqrt{p(1-p)} \in O \left( \sqrt{p} \right)$ as $p \rightarrow 0$,
  we can define a bucket for $t \in T$ as $J_{t,\rho}=[t-\rho\sqrt{t},t+\rho\sqrt{t}]$,
  where $\rho>0$ is chosen such that $0 \notin J_{t,\rho}$.
  \item The buckets can be chosen to match the precision of a na\"ive sampling method which draws a fixed number of samples $n \in \N$ per hypothesis.
  For this we compute all $n+1$ possible confidence intervals
  (one for each possible $S_n \in \{0,\ldots,n\}$)
  for each threshold $t \in T$ and record all confidence intervals which cover $t$.
  The union of those intervals can then be used as a bucket for $t$.
\end{enumerate}
The tuning parameter $\rho$ can be chosen, for instance, to minimize the maximal (worst case) effort of the resulting overlapping buckets.

The article leaves scope for a variety of future research directions.
For instance, how can the overlapping p-value buckets be chosen to maximize the probability of obtaining a classical decision (*, ** or ***),
subject to a suitable optimization criterion?
How can the lower bound on the computational effort derived in Section~\ref{sec:lowerbound} be improved?
Which algorithm (possibly based on our generic algorithm) is capable of meeting the effort of the lower bound?


\appendix
\section*{Appendix}

\section{Proofs}
\label{section_proofs}

\begin{proof}[Proof of Theorem~\ref{theorem_existance}]
We prove a circular equivalence of the three statements.

$(1.) \Rightarrow (2.)$: 
Suppose the buckets $\cal J$ are not overlapping.
This implies that there exists $\alpha\in (0,1)$ which is not contained in the interior of any $J\in {\cal J}$.
Let $I\in {\cal J}$ be the (random) interval reported by algorithm $A$ which satisfies \eqref{eq:RR}.
Let $n\in \N$ such that $\alpha-1/n\geq 0$ and $\alpha+1/n\leq 1$.

Consider the hypotheses $H_0: p=\alpha-1/n$ and $H_1: p=\alpha+1/n$
and the test that rejects $H_0$ iff $\alpha-1/n\notin I$.  As $I$
cannot contain both $\alpha-1/n$ and $\alpha+1/n$ (otherwise $\alpha$
would be in the interior of the interval $I$) and because of \eqref{eq:RR}, this
test has type I and type II error of at most $\epsilon$.  Hence, by the
lower bound on the expected number of steps of a sequential test given
in \cite[eq.~(4.81)]{wald1945}, see also \cite[section~3.1]{Gandy2009}, we have
$$
\E_{\alpha+1/n}(\tau)\geq \frac{\epsilon\log\left(\frac{\epsilon}{1-\epsilon}\right)+(1-\epsilon)\log\left(\frac{1-\epsilon}{\epsilon}\right)}
{\left(\alpha+\frac{1}{n}\right)\log\left(\frac{\alpha+1/n}{\alpha-1/n}\right)+
\left(1-\alpha-\frac{1}{n}\right)\log\left(\frac{1-\alpha-1/n}{1-\alpha+1/n}\right)}.
$$
As $n\to\infty$, the right hand side converges to $\infty$, contradicting (1.).

$(2.) \Rightarrow (3.)$:
We construct an explicit (but not very efficient) algorithm for this.

Let $a_0<a_1<\cdots<a_k$ be the ordered boundaries of the buckets in ${\cal J}$, i.e.\
$\{a_0,\dots,a_k\}=\{\max J: J\in{\cal J}\}\cup \{\min J: J\in{\cal J}\}$.
Let $\Delta=\min\{a_{i}-a_{i-1}:i=1,\dots,k\}$ be the minimal gap between those boundaries.

Let $I(S,n)$ be the two-sided \cite{Clopper1934} confidence interval
with coverage probability $1-\epsilon$ for $p$, where $n\in \N$ is the number
of samples and $S$ is the number of exceedances observed among those $n$ samples.
Let $n$ be such that the length of all Clopper-Pearson intervals is less than $\Delta$,
i.e.\ $n=\min \{ m\in \N: |I(S,m)|<\Delta~\text{for all}~S\in \{0,\dots,m\} \}$.
This is well-defined as the length of the Clopper-Pearson confidence
interval $I(S,n)$ decreases to 0 uniformly in $S$ as $n\to \infty$
(see e.g.\ the proof of Condition~2 in Lemma~2 of \cite{GandyHahn2014}).

Consider the algorithm that takes $n$ samples $X_1,\dots,X_n$ and then
returns an arbitrary interval $I\in {\cal J}$ that satisfies
$I\supseteq I(\sum_{i=1}^nX_i,n)$ (to be definite, order all elements
in ${\cal J}$ arbitrarily and return the first element satisfying the
condition). Such an $I$ always exists as the buckets are overlapping by (2.)\
and as $|I(\sum_{i=1}^nX_i,n)|<\Delta$, implying that it overlaps with at
most one possible boundary.
This algorithm satisfies \eqref{eq:RR} due to the coverage probability of $1-\epsilon$ of the Clopper-Pearson interval.

$(3.) \Rightarrow (1.)$:
Since finite effort implies expected finite effort, (1.) follows immediately.
\end{proof}

\begin{proof}[Proof of Lemma~\ref{le:finitestop}]
We first prove that the length of $\CnRL$ uniformly goes to zero.
The bounded stopping time then follows after proving that once an interval is below a certain length,
it is guaranteed to be contained in one of the buckets.

If $0\leq p \leq S_n/n-\left[ \log((n+1)/\epsilon)/(2n) \right]^{1/2}$ then, by Hoeffding's inequality \citep{Hoeffding1963},
$$b(n,p,S_n) = \PP(X=S_n) \leq \PP\left(\frac{X}{n}-p\geq \frac{S_n}{n}-p\right) \leq \exp \left( \frac{-2(S_n-np)^2}{n} \right) \leq \frac{\epsilon}{n+1},$$
where  $X \sim \text{Binomial}(n,p)$.
Hence,  $p \notin \CnRL$.

A similar argument shows that $b(n,p,S_n) \leq \epsilon/(n+1)$ for
$S_n/n+\left[ \log((n+1)/\epsilon)/(2n) \right]^{1/2} \leq p \leq 1$.
Thus, $|\CnRL| \leq \left[ 2 \log( (n+1)/\epsilon )/n \right]^{1/2}$.

Now assume no $c>0$ exists such that any interval $I\subseteq [0,1]$ with length less than $c$ is contained in a $J\in {\cal J}$.
Then for all $n \in \N$ there exists an interval $\CnRL \subset [0,1]$ with $0<|\CnRL|<1/n$ such that $\CnRL \not\subseteq J$ for all $J\in {\cal J}$.
Let $a_n$ be the mid point of $\CnRL$. As $(a_n)$ is a bounded sequence, there exists a convergent subsequence $(a_{n_k})$. Let $b=\lim_{k\to\infty}a_{n_k}$.

If $b\in (0,1) $ then, as ${\cal J}$ is overlapping, there exists
$\epsilon>0$ and $J\in {\cal J}$ such that $(b-\epsilon,b+\epsilon)\subseteq J$.
For large enough $k$  we have  $C_\text{RL}(X_{1:n_k}) \subseteq (b-\epsilon, b+\epsilon)$, contradicting $C_\text{RL}(X_{1:n_k}) \not \subseteq J$.

If $b=0$ then, as ${\cal J}$ is a covering of $[0,1]$ consisting of intervals of positive length, there exists $\epsilon>0$ and
$J\in {\cal J}$ such that $[0,\epsilon)\subseteq J$.
For large enough  $k$ we have $C_\text{RL}(X_{1:n_k}) \subseteq [0,\epsilon)$, again contradicting $C_\text{RL}(X_{1:n_k}) \not \subseteq J$.
If $b=1$, a contradiction can be derived similarly.
\end{proof}

\begin{proof}[Proof of Theorem~\ref{theorem_2epsilonNEW}]
\begin{enumerate}[wide]
  \item For threshold $\alpha\in B_{\cal J}$, let
  $\overline{E}_\alpha^{N} = \left\{ S_{\tau_\alpha} \geq U_{\tau_\alpha,\alpha}, \tau_\alpha < N \right\}$
  be the event that the upper boundary is hit first before time $N$
  and let
  $\underline{E}_\alpha^{N} = \left\{ S_{\tau_\alpha} \leq L_{\tau_\alpha,\alpha}, \tau_\alpha < N \right\}$
  be the event that the lower boundary is hit first.
  Then, for all $\alpha,\alpha' \in B_{\cal J}$ with $\alpha<\alpha'$,
  \begin{align}
  \overline{E}_{\alpha}^N \supseteq \overline{E}_{\alpha'}^N \quad\text{and}\quad
  \underline{E}_\alpha^N \subseteq \underline{E}_{\alpha'}^N.
  \label{eqn:monotonicity_boundaries}
  \end{align}
  Indeed, to see
  $\overline{E}_{\alpha}^N \supseteq \overline{E}_{\alpha'}^N$, we can argue as follows.
  On the event $\overline{E}_{\alpha'}^N$, as  $U_{n,\alpha} \leq U_{n,\alpha'}$ for all $n \in \N$,
  the trajectory $(n, S_n)$ must hit the upper boundary $U_{n,\alpha}$ of $\alpha$ no later than $\tau_{\alpha'}$,
  hence $\tau_{\alpha} \leq \tau_{\alpha'}<N$.
  It remains to prove that the trajectory does not first hit the lower boundary $L_{n,\alpha}$ of $\alpha$.
  Indeed, if the trajectory does hit the lower boundary of $\alpha$ before hitting its upper boundary,
  it also hits the lower boundary of $\alpha'$ (as $L_{n,\alpha} \leq L_{n,\alpha'}$ for all $n <N$)
  before time $\tau_{\alpha'}$,
  thus contradicting being on the event $\overline{E}_{\alpha'}^N$.
  Hence, we have $\overline{E}_\alpha^N \supseteq \overline{E}_{\alpha'}^N$.
  The proof of $\underline{E}_\alpha^N \subseteq \underline{E}_{\alpha'}^N$ is similar.

  Using this notation, for all $p\in [0,1]$,
  \begin{align}
  \nonumber
  \PP_p(\text{there exists}~n< N: p\notin \CnS)
  \nonumber
  &\leq \PP_p(\text{there exist}~n< N, \alpha \in B_{\cal J}: p\notin I_{n,\alpha})\\
  \nonumber
  &=\PP_p\left( \bigcup_{\alpha\in B_{\cal J}:\alpha\leq p}\underline{E}_\alpha^N \cup
    \bigcup_{\alpha\in B_{\cal J}:\alpha\geq p}\overline{E}_\alpha^N\right)\\
  &\leq \PP_p\left( \bigcup_{\alpha\in B_{\cal J}:\alpha\leq p}\underline{E}_\alpha^N\right)
    +\PP_p\left( \bigcup_{\alpha\in B_{\cal J}:\alpha\geq p}\overline{E}_\alpha^N\right).
  \label{le:splithitting}
  \end{align}
  If $p < \min B_{\cal J}$, the first term is equal to $0$.
  Otherwise, let $\alpha'=\max\{\alpha\in B_{\cal J}:\alpha\leq p\}$.
  Then, by \eqref{eqn:monotonicity_boundaries},
  $$
  \PP_p\left(\bigcup_{\alpha\in B_{\cal J}:\alpha\leq p}\underline{E}_\alpha^N\right)
  =\PP_p\left(\underline{E}_{\alpha'}^N\right)\leq\rho.
  $$
  The second term on the right hand side of \eqref{le:splithitting} can be dealt with similarly.
  
  \item By \eqref{eq:step1} and as $\Delta_n = o(n)$ there exists $n_0 \in \N$ such that
  \begin{equation}
    \label{eq:boundundecided}
    |\{\alpha \in B_{\cal J}: \tau_{\alpha} > n_0\}|\leq 1.
  \end{equation}
  We will show that $\tau_{A({\cal J},C_\text{S})} \leq n_0$.
  First, the assumption on the ordering of $L_n$ and $U_n$ excludes the possibility that $C_\text{S}(X_{1:n_0})=\emptyset$.
  Second, \eqref{eq:boundundecided} implies $|C_\text{S}(X_{1:n_0}) \cap B_{\cal J}|\leq 1$.

  If $|C_\text{S}(X_{1:n_0}) \cap B_{\cal J}|= 1$ then let $\alpha \in B_{\cal J}$ be such that $\alpha \in C_\text{S}(X_{1:n_0})$.
  As ${\cal J}$ is overlapping, there exist $J\in {\cal J}$ such that $\alpha$ is in the interior of $J$. 
  Hence, $\alpha$ cannot be a boundary of $J$, implying $C_\text{S}(X_{1:n_0}) \subseteq J$ due to $|C_\text{S}(X_{1:n_0}) \cap B_{\cal J}|=1$,
  thus showing $\tau_{A({\cal J},C_\text{S})} \leq n_0$.

  If $|C_\text{S}(X_{1:n_0}) \cap B_{\cal J}|= 0$ then let $\beta$ be in the interior of $C_\text{S}(X_{1:n_0})$.
  As ${\cal J}$ is overlapping, there exists $J\in {\cal J}$ such that $\beta\in J$.
  As $C_\text{S}(X_{1:n_0}) \cap B_{\cal J}=\emptyset$ this implies $C_\text{S}(X_{1:n_0}) \subseteq J$, thus showing $\tau_{A({\cal J},C_\text{S})} \leq n_0$.\qedhere
\end{enumerate}
\end{proof}

\begin{proof}[Proof of Lemma~\ref{le:eventual_ineqbounds}]
By arguments in \cite[Proof of Theorem 1]{Gandy2009}, we have 
\begin{align}
\frac{U_{n,\alpha} - n \alpha}{n} \leq \frac{\Delta_n+1}{n} \rightarrow 0,\qquad
\frac{L_{n,\alpha'} - n \alpha'}{n} \geq -\frac{\Delta_n+1}{n} \rightarrow 0,
\label{eq:step1}
\end{align}
as $n \rightarrow \infty$, where $\Delta_n = \sqrt{-n \log (\epsilon_n-\epsilon_{n-1})/2}$.
Since $\Delta_n = o(n)$ there exists $n_0 \in \N$ such that
\begin{equation}
\label{eq:def_no}
2 \left( \frac{\Delta_n}{n} + \frac{1}{n} \right) \leq \alpha' - \alpha \text{ for all }n \geq n_0. 
\end{equation}
Splitting $\frac{2}{n} = \frac{1}{n}+\frac{1}{n}$ and multiplying by $n$ yields
$n\alpha + \Delta_n + 1 \leq n\alpha' - \Delta_n - 1$
from which $U_{n,\alpha} \leq L_{n,\alpha'}$ follows by \eqref{eq:step1}.

By definition, we have $L_{n,\alpha} \leq U_{n,\alpha}$ and
$L_{n,\alpha'} \leq U_{n,\alpha'}$ for all $n\in\N$, thus implying
$L_{n,\alpha} \leq L_{n,\alpha'}$ and $U_{n,\alpha} \leq U_{n,\alpha'}$
for all $n \geq n_0$ as desired.
\end{proof}

\begin{proof}[Proof of Theorem~\ref{th:lowerbounds}]
We suppose that $I\in {\cal J}$ is the (random) bucket reported by a sequential
algorithm that respects \eqref{eq:RR}.  Let $\tilde p\in [0,1]$.  For any
$q\in [0,1]\setminus \tilde J$, we can consider the hypotheses
$H_0: p=\tilde p$ against $H_1:p=q$ and the test that rejects $H_0$
if and only if $\tilde p\notin I$.  By \eqref{eq:RR}, the type I error of such a test is
at most $\epsilon$.  Also, the type II error is at most $\epsilon$, as $q\notin \tilde J$ implies
$\PP_q(\tilde p\in I)\leq \PP_q(q\notin I)\leq \epsilon$.
Hence, using the lower bound in \cite[eq.~(4.80)]{wald1945}, we get \eqref{eq:lowerbound1}.

To see \eqref{eq:lowerbound2}: 
For any
$q\in [0,1]\setminus J_1$ consider the hypotheses $H_0:p=\tilde p$
and $H_1:p=q$ and the  test that rejects $H_{0}$ if and only if
$I\neq J_1$. This test has type I error $1-\eta$, where $\eta=\PP_{\tilde p}(I= J_{1})$, and type II error of at most
$\epsilon$.  Using \cite[eq.~(4.80)]{wald1945} we get
$\E_{\tilde p}(\tau) \geq  e(\tilde p, q,1-\eta,\epsilon)$.  Similarly, for any $q\in [0,1]\setminus J_2$, we
can test the hypotheses $H_0:p=\tilde p$ and $H_1:p=q$ by rejecting
$H_{0}$ if and only if $I\neq J_2$. This test has type I error of at most
$\min(\eta+\epsilon,1)$ and type II error of at most $\epsilon$.  Again, using
\cite[eq.~(4.80)]{wald1945} we get 
$\E_{\tilde p}(\tau)\geq e(\tilde p, q,\min(\eta+\epsilon,1),\epsilon)$.
Eq.~\eqref{eq:lowerbound2} follows as these inequalities hold for all $q$ and due to the fact that we can account for the unknown $\eta$ by minimizing over it.
\end{proof}

\section{A simple stopping criterion for Robbins-Lai}
\label{section_simple_criterion_RL}
The following describes a simple criterion to determine whether a confidence interval computed via the Robbins-Lai approach of Section~\ref{sec:robbinslai}
is fully contained in a bucket. For a single threshold this approach has been suggested in \cite{Ding2016}.
Let interval $\CnRL$ and bucket $J \in {\cal J}$ as well as $n$, $S_n$ and $\epsilon$ be as in Sections~\ref{section_general_algorithm} and \ref{sec:robbinslai}.
Then $\CnRL \subseteq J$ if and only if for $p \in \{\min J,\max J\}$,
\begin{align}
(n+1) b(n,S_n,p) = (n+1) \binom{n}{S_n} p^{S_n} (1-p)^{n-S_n} \leq \epsilon.
\label{eq:conf}
\end{align}
As \eqref{eq:conf} is also satisfied if $\CnRL$ and $J$ are simply disjoint,
we verify that $(n+1) b(n,S_n,p)$ is indeed increasing at $\min J$ and decreasing at $\max J$ using the derivative of $(n+1) b(n,S_n,p)$ with respect to $p$.

After applying a (monotonic) log transformation to \eqref{eq:conf}, taking the derivative with respect to $p$ yields
\begin{align}
\frac{S_n}{p} - \frac{n-S_n}{1-p}
\begin{cases}
\geq 0 & p=\min J,\\
\leq 0 & p=\max J.
\end{cases}
\label{eq:deriv}
\end{align}
If \eqref{eq:conf} and \eqref{eq:deriv} are satisfied, then $\CnRL \subseteq J$.

\end{document}